\documentclass[11pt]{article}
\usepackage{hyperref}
\usepackage{fullpage}
\usepackage{amsmath}
\usepackage{latexsym}
\usepackage{ifthen}

\usepackage{amsfonts}
\usepackage{amsthm}

\newcommand{\be}{\begin{eqnarray} \begin{aligned}}
\newcommand{\ee}{\end{aligned} \end{eqnarray} }
\newcommand{\benn}{\begin{eqnarray*} \begin{aligned}}
\newcommand{\eenn}{\end{aligned} \end{eqnarray*} }

\newcommand{\bc}{\begin{center}}
\newcommand{\ec}{\end{center}}

\newcommand{\id}{\mathbb{I}}

\newcommand{\tr}{\mathop{\mathrm{Tr}}\nolimits}

\newcommand{\norm}[1]{\left\| #1\right \|}

\newtheorem{theorem}{Theorem}[section]

\newtheorem{lemma}[theorem]{Lemma}
\newtheorem{definition}[theorem]{Definition}

\newcommand{\hil}{\mathcal{H}}

\usepackage{amsfonts}

\def\id{\mathbb{I}}

\def\01{\{0,1\}}

\newcommand{\eps}{\varepsilon}
\newcommand{\ket}[1]{|#1\rangle}

\newcommand{\outp}[2]{|#1\rangle\langle#2|}

\newcommand{\mH}{\mathcal H}

\newcommand{\mX}{\mathcal X}

\newcommand{\mA}{\mathcal A}
\newcommand{\mB}{\mathcal B}

\newcommand{\mR}{\mathcal R}

\newcommand{\mE}{\mathcal E}

\newcommand{\bA}{\mathbf A}

\newcommand{\bF}{\mathbf F}
\newcommand{\bG}{\mathbf G}
\newcommand{\bH}{\mathbf H}

\newcommand{\bP}{\mathbf P}
\newcommand{\bQ}{\mathbf Q}

\newcommand{\bS}{\mathbf S}

\newcommand{\cancel}[1]{}

\newcommand{\PlayerA}{{\textsf{A}}}
\newcommand{\PlayerB}{{\textsf{B}}}

\newcommand{\regQ}{\mathcal{Q}}
\newcommand{\regA}{\mathcal{A}}
\newcommand{\regM}{\mathcal{M}}
\newcommand{\regK}{\mathcal{K}}

\newcommand{\regE}{\mathcal{E}}
\newcommand{\regX}{\mathcal{X}}
\newcommand{\regB}{\mathcal{B}}

\newcommand{\regC}{\mathcal{C}}
\newcommand{\regR}{\mathcal{R}}
\newcommand{\regY}{\mathcal{Y}}

\newcommand{\setS}{\mathbb{S}}
\newcommand{\setT}{\mathbb{T}}

\newcommand{\setX}{\mathbb{X}}
\newcommand{\setY}{\mathbb{Y}}
\newcommand{\setZ}{\mathbb{Z}}

\newcommand{\setA}{\mathbb{A}}

\DeclareMathOperator{\inn}{in}
\DeclareMathOperator{\outt}{out}

\newcommand{\ROT}[3]{{#2 \choose #1}{\textsf{-ROT}^{#3}}}
\newcommand{\ROTT}{{\textsf{ROT}}}

\newcommand{\TOR}[3]{{#2 \choose #1}{\textsf{-TOR}^{#3}}}
\newcommand{\TORR}{{\textsf{TOR}}}

\newcommand{\BC}{{\textsf{BC}}}

\newcommand{\OT}[3]{{#2 \choose #1}{\textsf{-OT}^{#3}}}
\newcommand{\OTT}{{\textsf{OT}}}
\newcommand{\OTfromROT}{\textsf{OTfromROT}}

\newcommand{\comm}{{\textsf{Comm}}}
\newcommand{\qComm}{{\textsf{Q-Comm}}}
\newcommand{\bqsOT}{{\textsf{BQS-OT}}}

\newcommand{\bqsTO}{{\textsf{BQS-TO}}}

\newcommand{\OTtoBC}{{\textsf{OTtoBC}}}

\newcounter{protoCount}
\newcounter{protoList}
\newsavebox{\tmpbox}
\newlength{\protobox}
\newenvironment{protocol}[2]{
\bigskip
\addtocounter{protoCount}{1}
\noindent \begin{lrbox}{\tmpbox}
\setlength{\protobox}{\textwidth}
\addtolength{\protobox}{-0.5cm}
\begin{minipage}[c]{\protobox}
\begin{bfseries}Protocol \theprotoCount: #1\end{bfseries}
\ifthenelse{\equal{#2}{\empty}}{}{\\Prerequisite: #2}
\begin{list}{\begin{bfseries}\arabic{protoList}:\end{bfseries}}
{\usecounter{protoList}}
}{
\end{list}
\end{minipage}\end{lrbox}
\fbox{\usebox{\tmpbox}}
\bigskip
}

\begin{document}

\title{{\sf Composable Security in the Bounded-Quantum-Storage Model}}

\author{Stephanie Wehner
\thanks{Supported 
by EU fifth framework project QAP IST 015848 and the NWO vici project 2004-2009.}
\\CWI, Amsterdam\\wehner@cwi.nl
\and J{\"u}rg Wullschleger\\University of Bristol\\juerg@wulli.com}
\date{\today}
\maketitle

\begin{abstract}
We present a simplified framework for proving sequential composability in the quantum setting. 
In particular, we give a new, simulation-based, definition for security in the bounded-quantum-storage model,
and show that this definition allows for sequential composition of protocols. 
Damg{\aa}rd \emph{et al.} (FOCS '05, CRYPTO '07) showed how to securely implement bit commitment and
oblivious transfer in the bounded-quantum-storage model, where the
adversary is only allowed to store a limited number of qubits.
However, their security definitions did only apply to the standalone setting, and
it was not clear if their protocols could be composed.
Indeed, we first give a simple attack that shows that these protocols are \emph{not} composable
without a small refinement of the model.
Finally, we prove the security of their randomized oblivious transfer protocol in our refined model. 
Secure implementations of oblivious transfer and bit commitment then follow easily by a
(classical) reduction to randomized oblivious transfer.
\end{abstract}

\section{Introduction}

Secure two-party computation~\cite{Yao82} allows two mutually distrustful players to jointly compute the value of
a function without revealing more information about their inputs than can be inferred from the function value itself.
In this context, the primitives known as bit commitment (BC) \cite{blum:coin}
and oblivious transfer (OT) \cite{Wiesner70,Rabin81,EvGoLe85} are of
particular importance: \emph{any} two-party computation can be implemented, provided these two primitives are available \cite{GolVai87,Kilian88,crepeau:committedOT}.

In bit commitment, 
the committer (Alice) secretly chooses a bit $b$, and commits herself to $b$ by exchanging messages with the verifier (Bob).
From the commitment alone, Bob should not be able to gain any information about $b$. Yet, when Alice later reveals
$b$ and opens the commitment by exchanging messages with Bob, he can verify whether Alice is truthful and
had indeed committed herself to $b$. 
In oblivious transfer, the sender (Alice) chooses two bits $x_0$ and $x_1$, 
the receiver (Bob) chooses a bit $c$. The protocol of oblivious transfer allows Bob to retrieve $x_c$
in such a way that Alice cannot gain any information about $c$. At the same time, Alice can be ensured 
that Bob only retrieves $x_c$ and no information about the other input bit $x_{1-c}$.

Unfortunately, BC and OT are impossible to implement securely without any additional assumptions, 
even in the quantum model~\cite{Mayers97,LoChau97}. This result holds even in the presence
of the so-called superselection rules~\cite{kitaev:super}.
Exact tradeoffs on how well we can implement BC in the quantum world can be found in~\cite{spekkens:tradeoffBc}.
To circumvent this problem (in both, the classical and the quantum case), we thus need to assume that the adversary is limited.
In the classical case, one such limiting assumption
is that the adversary is \emph{computationally bounded}, i.e., he is restricted to a polynomial time computations 
(see e.g. \cite{Naor91,EvGoLe85}).
In the quantum model, it is also possible to securely implement both protocols provided that an adversary cannot measure more than a fixed number of qubits
simultaneously~\cite{salvail:physical}. Very weak forms of string commitments can also be obtained~\cite{wehner06d}.

\paragraph{The Bounded-Quantum-Storage Model.}

Of particular interest to us is the bounded-storage model. Here, the adversary is bounded in \emph{space}
instead of time, i.e., she is only allowed to use a certain amount of storage space. Both OT and BC can be implemented
in this model \cite{cachin:boundedOT}.
Yet, the security of a \emph{classical} bounded-storage 
model~\cite{Maurer92b,cachin:boundedOT} is somewhat unsatisfactory:
First, a dishonest player needs only quadratically more memory than the honest one.
Second, as classical memory is very cheap, most of these protocols require a huge amount of
communication in order to achieve reasonable bounds on the adversaries memory.
In the quantum case, on the other hand, it is \emph{very} difficult to store states even for a very short period of time. This
leads to the protocol presented in \cite{BBCS92,Crepea94}, which show how to implement BC and OT if the adversary is
not able to store \emph{any} qubits at all.
In~\cite{serge:bounded,serge:new}, these ideas have been generalized in a very nice way
to the \emph{bounded-quantum-storage model}, where the adversary is computationally unbounded and 
allowed to have an unlimited amount of \emph{classical} memory. However, he is only allowed 
a limited amount of \emph{quantum} memory. 
The advantages over the classical bounded-storage model are two fold: 
First, given current day technology it is indeed very hard to store quantum states. Secondly, here the honest
player does not require any quantum storage at all, making the protocol implementable using present day technology.

\paragraph{Security Definitions and Composability.}

Cryptographic protocols (especially protocols that implement very basic functionalities such as BC or OT) are almost never executed on their own. They are merely used
as building blocks for larger, more complicated applications. However, it is not clear that
the composition of secure protocols will remain secure. 
Formal security definitions for secure function evaluation have first been proposed in \cite{MicRog91} and \cite{Beaver91}. These definitions use the \emph{simulation paradigm} invented in \cite{GoMiRa89} to define zero-knowledge proofs of knowledge.
In \cite{Canetti00b} it has been shown formally that these definitions imply
that protocols can be composed \emph{sequentially}.
Sequential composition implies that protocols can be composed in an arbitrary way, as long as
at any point in time exactly one protocol is running. All other protocols have to wait
until that protocol stops.
A stronger security definition called \emph{universal composability} has been introduced in
\cite{Canetti00,PfiWai00,BaPfWa03}. It guarantees that protocols can
be securely composed in an arbitrary way (also concurrently) in any environment.

Simulation-based security requires that for any adversary attacking the real
protocol there exists a simulator in the ideal setting, i.e. where the players only have black-box access to an ideal
functionality, such that the environment cannot distinguish between the real
and the ideal setting. To make the protocol sequentially composable, we have to allow the adversary to receive some \emph{auxiliary input} from 
the environment, which could contain information from a previous run of the protocol, the larger application that the protocol is embedded in, 
or any other information that the environment might pass to the adversary in an attempt to distinguish between the real from the ideal setting.
In the quantum case, this auxiliary input is an arbitrary quantum state, unknown to the adversary or the simulator.
This presents us with two additional difficulties we do not encounter in the 
classical setting: First, the simulator cannot determine what this input state is without disturbing the state, which could be 
detected by the environment. Second, the input state may be entangled with the environment. 
Based on earlier work in \cite{vGraaf98}, a simulation-based framework for secure quantum multi-party computation has been presented in \cite{Smith01}. That framework offers sequential composability, however
no composability theorem was presented there. Universal composability in the quantum world has been introduced in~\cite{benor:compose},  and independently in \cite{unruh:compose}.
Their framework is very powerful. But, due to their complexity, hard to apply.

In \cite{Unruh02}, it has been shown that classical protocols which have been proven to be universally composable using their \emph{classical} definitions, are secure against \emph{quantum} adversaries. This is result is very useful, as it allows us to use many classical protocols also in the quantum setting.
Great care must be taken in the definition of security in the quantum setting:
For example,
the standard security definition for QKD based on accessible information does
not imply composability~\cite{KRBM05}.

\subsection{Contribution}

In \cite{serge:new}, protocols for OT and BC have been presented and shown to be secure against bounded quantum adversaries. However, the proofs only guarantee security in a standalone setting. Indeed, a very simple attack shows that they are not composable.
The main contribution of this paper is to give a formal framework for security in the bounded-quantum-storage model, and to show that modified versions of these protocols are \emph{sequentially composable}. Hence, they can be used as building blocks in other protocols.

\paragraph{Proofs in \cite{serge:bounded,serge:new} do not imply Composability.}

When considering composable security, we need to allow the adversary to receive some auxiliary quantum input.
This has not been considered in the security definitions used in
\cite{serge:bounded,serge:new}. When we allow this, we are faced with two major
problems:
First, in the security proof of~\cite{serge:new}, the receivers choice bit
can only be extracted by the simulator if the distribution of the senders random string 
given the receivers classical knowledge is known, which is not the case if the adversary has auxiliary input.
Second, the \emph{memory bound} is only enforced at \emph{one} specific step in their protocol,
while during the rest of the protocol, the adversary is
allowed unlimited memory. The following very simple EPR-attack shows how the protocol can then be broken:
We let the adversary receive an arbitrary number of halves of EPR-pairs from the environment 
as his auxiliary input, and run the protocol as before. Then, just before the memory bound is applied, he \emph{teleports} his whole
quantum memory to the environment. The classical communication needed to teleport can be part of the adversaries classical storage
that he later outputs. Thus, the adversary can artificially increase his own storage by borrowing some quantum memory from the environment.

One possibility to overcome the second problem is to limit the memory of the environment. 
Yet, this solutions seems very unsatisfactory: While we may be willing to accept that, say, a 
smart-card cannot store more than 100 qubits, this is much less clear for the environment. How could we place any limitations 
on the environment at all?
In our framework, we thus always allow the environment to have an arbitrary amount of quantum memory,
but limit the adversaries memory.

\paragraph{Composable Security in the Bounded-Quantum-Storage Model.}

We start by presenting a formal model for secure two-party computation in the bounded-quantum-storage model. 
Our model is quite similar to the model presented in \cite{Smith01}, and provides \emph{offline-security}.
Then, we show that our model implies that secure protocols are sequentially composable.

Second, we slightly modify the model from \cite{serge:new}
and prove the security of randmomized OT in our refined model, which implies that the protocol is composable.
In particular, we introduce a second memory bound into the protocol, which limits the amount of quantum 
auxiliary input the adversary may receive. We show that the simulator can extract the choice bit, even
if the auxiliary quantum input remains completely unknown to him, and that the protocol is secure even if the quantum memory
of the environment is unbounded. It turns out that the protocol only remains secure for a smaller memory bound in our model.

Third, we give well-known \emph{classical} reductions of BC and OT to
randomized OT in the appendix, and prove that they are secure in our model.
Using the idea from \cite{Beaver95}, this also implies that the two players can \emph{precompute} $\ROTT$,
and, at a later point in time, they can use it to implement either an OT or a BC,
for which they only need classical communication.

Since the proof presented in \cite{Unruh02} carries over to our model,
secure function evaluation in the bounded-quantum-storage model can be achieved by simply using
the (classical) universal composable protocols presented in \cite{Estren04}, which are based on \cite{crepeau:committedOT}. Note that because our implementation
of OT is physical, the results presented in \cite{CLOS02} cannot be applied.

\paragraph{Outline}
In Section~\ref{prelim}, we introduce the basic tools that we need later. In Section~\ref{qframework}, we define a framework
that provides offline security in the bounded-quantum-storage model, which implies that protocols
can be composed sequentially.
In Section~\ref{OTsecurity}, we then prove the security of the randomized oblivious transfer
protocol from \cite{serge:new} in our refined model.
In the appendix, we show that secure implementations of oblivious transfer and bit commitment follow by a
(classical) reduction to randomized oblivious transfer.

\section{Preliminaries}\label{prelim}

\subsection{Notation}
We assume general familiarity with the quantum model~\cite{hayashi}. Throughout this paper,
we use the term \emph{computational basis} to refer to the basis given by $\{\ket{0},\ket{1}\}$.
We write $+$ for the computational basis, and let $\ket{0}_+ = \ket{0}$ and $\ket{1}_+ = \ket{1}$.
The \emph{Hadamard basis} is denoted by $\times$, and given by $\{\ket{0}_\times,\ket{1}_\times\}$, 
where $\ket{0}_\times = (\ket{0} + \ket{1})/\sqrt{2}$
and $\ket{1}_\times = (\ket{0}-\ket{1})/\sqrt{2}$. 
For a string $x \in \01^n$ encoded in bases $b \in \{+,\times\}^n$, we write $\ket{x}_b = \ket{x_{1}}_{b_1},\ldots,
\ket{x_{n}}_{b_n}$.
We also use 0 to denote $+$, and 1 to denote $\times$.
Finally, we use $x_{|c}$ to denote the sub-string of $x$ consisting of all $x_i$ where $b_i = c$.

We use the font $\mA$ to label a quantum register, corresponding
to a Hilbert space $\mA$. A \emph{quantum channel} from $\mA$ to $\mB$ 
is a completely positive trace preserving (CPTP) map $\Lambda: \mA \rightarrow \mB$.
We also call a map from $\mA$ to itself a \emph{quantum operation}.
Any quantum operation on the register $\mA$ can be phrased as a unitary operation
on $\mA$ and an additional ancilla register $\mA'$, where we trace out $\mA'$ to obtain the actions of the
quantum operation on register $\mA$~\cite{hayashi}. We use $\setS(\mA)$ to refer to the set of all quantum states in $\mA$,
and $\setT(\mA)$ to refer to the set of all Hermitian matrices in $\mA$.
We use ${\bf U}$ to refer to a quantum operation,
upper case letters $X$ to refer to classical random variables,
the font $\setS$ for a set,
and the font $\PlayerA$ to refer to a player in the protocol.

\subsection{Distance Measures}\label{quantumPrelim}

Our ability two distinguish to
quantum states is determined by their \emph{trace distance}. The trace distance between two states $\rho, \rho' \in \setS(\hil)$
is defined as
$
D(\rho,\rho') := \frac12 \tr  | \rho - \rho' |
$,
where  $|A| = \sqrt{A^\dagger A}$.
We also write $\rho \equiv_\eps \rho'$, if $D(\rho,\rho') \leq \eps$.
For all practical purposes, $\rho \equiv_\eps \rho'$ means that the state $\rho'$ behaves
like the state $\rho$, except with probability $\eps$~\cite{RenKoe05}.
For any quantum channel $\Lambda$, we have
$D(\Lambda(\rho),\Lambda(\rho')) \leq D(\rho,\rho')$.
Furthermore, the triangle inequality holds, i.e., for all $\rho$, $\rho'$ and $\rho''$, we have
$D(\rho,\rho'') \leq  D(\rho,\rho') + D(\rho',\rho'')$. 
Let $\Lambda,\Lambda' : \mA \rightarrow \mB$ be two quantum channels. If for all $\rho \in \mA$, we have have
$\Lambda(\rho) \equiv_\eps \Lambda'(\rho)$,
we may also write $\Lambda\equiv_\eps \Lambda'$.
Let $\rho_{A B} \in \setS(\mA \otimes \mB)$ be classical on $\mA$, i.e. $\rho_{AB} = \sum_{x\in \mX} P_X(x) \outp{x}{x} \otimes \rho_x$
for some distribution $P_X$ over a finte set $\mX$. We say that \emph{$A$ is $\eps$-close to
uniform with respect to $B$}, if $D(\rho_{AB}, \id_A / d \otimes \rho_B) \leq \eps$,
where $d = \dim(\mH_A)$.

\subsection{Uncertainty Relation and Privacy Amplification}

For random variables $X$ and $Y$ with joint
distribution $P_{X Y}$, the \emph{smooth conditional min-entropy}~\cite{RenWol05} 
can be expressed in terms of an optimization over events $\mE$ with probability at
least $1-\eps$. Let $P_{X \mE|Y=y}(x)$ be the probability that
$\{X=x\}$ \emph{and} $\mE$ occur conditioned on $Y=y$. We have
$$
H^\eps_{\min}(X|Y) = \max_{\mE: \Pr(\mE)\geq 1-\eps} \min_{y} \min_x (-\log P_{X \mE|Y=y}(x)).
$$
The smooth min-entropy allows us to use the following chain rule which does not hold in the case of standard min-entropy.
\begin{lemma}[Chain Rule \cite{Cachin97b,MauWol97,RenWol05}] \label{lem:chain}
Let $X$, $Y$, and $Z$ be arbitrary random variables over $\setX$, $\setY$ and $\setZ$. Then for all $\eps, \eps' > 0$,
$$
 H^{\eps+\eps'}_{\min}(X | YZ )  \geq H^\eps_{\min}(XY \mid Z) - \log|\setY| - \log(1/\eps').
$$
\end{lemma}
We also need the following monotonicity of the smooth min-entropy 
$$H^{\eps}_{\min}(XY \mid Z)  \geq H^\eps_{\min}(X \mid Z).$$

A function $h: \setS \times \setX \rightarrow \{0,1\}^\ell$ is called a \emph{two-universal hash function} \cite{CarWeg79}, if
for all $x_0 \neq x_1 \in \setX$, we have $\Pr [h(S,x_0) = h(S,x_1)] \leq 2^{-\ell}$
if $S$ is uniform over $\setS$.
We thereby say that a random variable $S$ is \emph{uniform over} a set $\setS$ if
$S$ is chosen from $\setS$ according to the uniform distribution.
For example, the class of all functions from $\setS \times \setX \rightarrow \{0,1\}^\ell$ is
two-universal. \emph{Privacy amplification} shows a two-universal hash function
can be used to extract an almost random string from
a source with enough min-entropy.
The following theorem is from \cite{serge:new}, stated slightly differently than
the original statements in \cite{RenKoe05,Renner05}.

\begin{theorem}[Privacy Amplification \cite{RenKoe05,Renner05}] \label{thm:QBS-ILL}
Let $X$ and $Z$ be (classical) random variables distributed over $\setX$ and $\setZ$, and let $Q$
be a random state of $q$ qubits.
Let $h: \setS \times  \setX \rightarrow \{0,1\}^\ell$ be a two-universal hash function and let
$S$ be uniform over $\setS$. If
$$
\ell \leq H^{\eps'}_{\min}(X \mid Z) - q -  2\log(1/\eps)
$$,
then $h(S,X)$ is $(\eps+2\eps')$-close to uniform with respect to $(S,Z,Q)$.
\end{theorem}

The following lemma follows from the uncertainty
relation presented in \cite{serge:new} by a simple purification argument and by fixing
the parameter $\lambda$ such that the error is at most $\eps$. (see Appendix)

\begin{lemma} \label{lem:uncertainty2}
Let $X \in \{0,1\}^n$ be a uniform random string, let $B \in \{+,\times\}^n$ be a uniform random basis.
Let $\ket{X}_B = (\ket{X_{1}}_{B_1},\ldots, \ket{X_{n}}_{B_n})$ be a state of $n$ qubits, and
let $K$ be the outcome of an arbitrary measurement of $\ket{X}_B$, which does not depend on
$X$ and $B$. Then, for any $\eps$, we have
$$H^{\eps}_{\min}(X | B K) \geq \frac n 2 - 10 \sqrt[3]{n^2  \log{\frac{1}{\eps}}},$$
which is positive if $n > 8000 \log(1 / \eps)$.
\end{lemma}

\section{Security in the Bounded-Quantum-Storage Model} \label{BQS-Model}\label{qframework}

We now give a definition of offline-security in the bounded-quantum-storage model,
and show that it allows protocols to be composed
\emph{sequentially}\footnote{Sequentially means that any given time only one sub-protocol is executed.}. 
Our definitions are closely related to~\cite{Smith01}.

First of all, we assume that there is a global clock, that divides time into discrete rounds.
We look at the following setting: Two \emph{players}, $\PlayerA$ and $\PlayerB$, execute a \emph{protocol} $\bP= (\bP_\PlayerA,\bP_\PlayerB)$, where $\bP_{\PlayerA}$ is the program
executed by $\PlayerA$ and $\bP_{\PlayerB}$ the program executed by $\PlayerB$. 
Before the first round, each program receives an input (that might be entangled with the input of the other player)
and stores it.
In each round, each program may 
first send/receive messages to/from a given functionality $\bG$, then
apply a quantum operation to its current internal storage (including the message space),
and finally send/receive further messages at the end of each round.
$\bG$ defines the communication resources available between the players, modeled as an interactive quantum functionality. It may contain a classical and/or a quantum communication channel, or other functionalities such as oblivious transfer or bit commitment. Finally, in the last step of the protocol each program outputs an output value.
Note that the execution of $\bP$ using $\bG$ --- denoted by $\bP(\bG)$ --- is a quantum channel, which
takes the input of both parties to the output of both parties.
We also use the term \emph{interface} of a player, to denote the interface presented by his program.

Players may be \emph{honest}, which means that they follow the protocol, or they may be \emph{corrupted}.
All corrupted players belong to the \emph{adversary}, $\setA \subset \{\PlayerA,\PlayerB\}$. We 
ignore the case where both players are corrupted, and we 
assume that this set is \emph{static}, i.e., it is already
fixed before the protocol starts.
We only consider the case where the adversary is
\emph{active}, i.e., the adversary may not follow the protocol. 
The adversary $\setA = \{p\}$ may replace his part of the protocol $\bP_p$ by another program $\bA_p$.
Opposed to $\bP_p$, $\bA_p$ receives some \emph{auxiliary (quantum) input} in the first round
that may also be entangled with the environment.
We do not restrict the computational power of $\bA_p$ in any way, however we do limit its internal quantum storage to a certain \emph{memory-bound} of $m$ qubits. We call such a $\bA_p$ \emph{$m$-bounded}. 
$\bA_p$ is allowed to perform arbitrary quantum operations in each round of the protocol. 
However after receiving his input, and after every round, all of his internal memory is measured, 
except for $m$ qubits.\footnote{Note that we enforce the memory bound after \emph{every} round to keep the model simple. Later, in the security proof of our randomized OT protocol, we see that the bound needs only to be enforced twice. A practical implementation may introduce a wait time at these points to make sure the quantum memory physically decoheres.}
$\bA_p$ is not allowed to input or output any additinal data \emph{during} the execution of the protocol.
The execution of $\bP$ using $\bG$, where $\bP_p$ has been replaced by $\bA_p$, is again a quantum channel, which maps the inputs of both players and the
auxiliary input of the adversary to the outputs produced by both programs.

The \emph{ideal functionality} defines what functionality we expect the protocol to implement.
For the moment we only consider  \emph{non-interactive} functionalities, i.e.,
both players can send it input only once at the beginning, and obtain the output only
once at the end. These functionalities have the form of a quantum channel.
To make the definitions more flexible, we allow $\bF$ to look differently depending on whether both players are honest, or either $\PlayerA$ or $\PlayerB$ belongs to
the adversary. So the ideal functionality is in fact a \emph{collection of functionalities},
$\bF = (\bF_{\emptyset}, \bF_{\{\PlayerA\}}, \bF_{\{\PlayerB\}})$.
$\bF_{\emptyset}$ denotes the functionality for the case when both players are honest,
 and $\bF_{\{\PlayerA\}}$ and
$\bF_{\{\PlayerB\}}$ for the cases when $\PlayerA$ or $\PlayerB$ respectively are dishonest.
We require that the honest player must always have the same interface as in $\bF_{\emptyset}$, i.e., in $\bF_{\{\PlayerA\}}$, $\PlayerB$ must have the same interfaces as in $\bF_{\emptyset}$, and in $\bF_{\{\PlayerB\}}$, $\PlayerA$ must have the same interfaces as in
$\bF_{\emptyset}$.
We also require that $\bF_{\{\PlayerA\}}$ and $\bF_{\{\PlayerB\}}$ allow the adversary to play honestly,
i.e., they must be at least as good for the adversary as the functionality $\bF_{\emptyset}$. 

We say that a protocol $\bP$ having access to the functionality $\bG$\footnote{$\bG$ may also be a collection
of functionalities.} securely
implements a functionality $\bF$, if the following conditions  are satisfied: First of all, we require that
the protocol has almost the same output as $\bF$, if both players are honest. Second, for $\setA = \{p\}$, we require that
the adversary 
attacking the protocol has basically no advantage over attacking
$\bF$ directly. We thus require that for every  $m$-bounded program $\bA_p$,
there exists a $s$-bounded program $\bS_p$ (called the \emph{simulator}), such that the overall outputs of
both situations is almost the same, for all inputs.
For simplicity, we do not make any restrictions on the efficiency of the simulators\footnote{Recall that
the adversary is computationally unbounded as well.}. 
Also, we do not require him to use the adversary $\bA_p$ as a black-box:
$\bS_{p}$ may be constructed
from scratch, under full knowledge of the behaviour of $\bA_p$.
In particular, we allow him to execute some or all actions of $\bA_p$ in a single round.
Recall, that a memory bound is applied only after each round.
Thus, when executing $\bA_p$ in a single round, the simulator will not experience any memory bound. This model is
motivated by the physically realistic assumption that such memory bounds are introduced by adding 
specific waiting times after each round. Hence, this does not give the simulator any memory.
However, in order to make protocols composable with other protocol
in our model,
we do require the simulator to be memory-bounded as well. 
The amount of memory
required by the simulator gives a bound on the \emph{virtual} memory the adversary seems to have by attacking the real protocol instead of the ideal one.
Ideally, we would like $\bS_{p}$ to use the same amount of memory as $\bA_p$.
The simulator $\bS_{p}$ can be represented by two quantum channels. The first channel maps the
input and the auxiliary input to an input to the ideal functionality, and to a state of at most $s$ qubit.
The other channel maps that state and the output of the ideal functionality to the output of the simulator.

\begin{definition} \label{def:bqm-sc-sec}
A protocol $\bP(\bF) = (\bP_{\PlayerA},\bP_{\PlayerB})(\bF)$ \emph{implements $\bG$ with an error of at most
$\eps$}, secure against $m$-bounded adversaries using $s$-bounded simulators, if
\begin{itemize}
\item(Correctness) $\bP(\bF_{\emptyset}) \equiv_\eps \bG_{\emptyset}$\;.
\item(Security for \PlayerA) For every $m$-bounded $\bA_{\PlayerB}$
there exists a $s$-bounded $\bS_{\PlayerB}$, such that
\[(\bP_{\PlayerA},\bA_{\PlayerB})(\bF_{\{\PlayerB\}}) \equiv_\eps \bS_{\PlayerB}(\bG_{\{\PlayerB\}})\;.\]
\item(Security for \PlayerB) For every $m$-bounded $\bA_{\PlayerA}$
there exists a $s$-bounded $\bS_{\PlayerA}$, such that
\[(\bA_{\PlayerA},\bP_{\PlayerB})(\bF_{\{\PlayerA\}}) \equiv_\eps \bS_{\PlayerA}(\bG_{\{\PlayerA\}})\;.\]
\end{itemize}
\end{definition}

An important property of our definition is that it allows protocols to be composed. The following theorem shows that in a secure protocol that is based on an ideal, non-interactive functionality $\bG$ and some other funtionalities $\bG'$\footnote{We denote the concatenation of the functionalities $\bG$ and $\bG'$ by $\bG \| \bG'$.}, we can replace $\bG$ with a secure implementation of $\bG$, without making the protocol insecure. The theorem requires that $\bG$ is called sequentially, i.e., that no other subprotocols are running parallel to $\bG$. 
The proof uses the same idea as in the classical case \cite{Canetti00b}.

\begin{theorem}[Sequential Composition Theorem] \label{thm:bqs-compose}
Let $\bF$ and $\bG$ be non-interactive functionalities, and $\bG'$ and $\bH$ be arbitrary functionalities.
Let $\bP(\bG \| \bG')$ be a protocol that calls $\bG$ sequentially and that implements $\bF$ with error of at most $\eps_1$ secure against $m_1$-bounded adversaries using $s_1$-bounded simulators,
and let $\bQ(\bH)$ be a protocol that implements $\bG$ with error of at most $\eps_2$
secure against $m_2$-bounded adversaries using $s_2$-bounded simulators, where $m_2 \geq s_1$. Then
$\bP(\bQ(\bH) \| \bG')$ implements $\bF$ with error at most $\eps_1 + \eps_2$,
secure against $\min(m_1,m_2)$-bounded adversaries using $s_2$-bounded simulators.
\end{theorem}

\begin{proof}(Sketch)
If both players are honest, the statement follows directly from the properties of the trace distance, since we
have $\bP(\bG_\emptyset \| \bG'_\emptyset) \equiv_{\eps_1} \bP(\bQ(\bH_\emptyset) \| \bG'_\emptyset)$, and hence
$\bF_\emptyset \equiv_{\eps_1 + \eps_2} \bP(\bQ(\bH_\emptyset) \| \bG'_\emptyset)$.

Let $\PlayerA$ be honest, and let $\PlayerB$ attack the protocol $\bP(\bQ(\bH) \| \bG')$ by
executing $\bA_\PlayerB$. We cut $\bA_\PlayerB$ into three parts.
Let $\bA^{(0)}_\PlayerB$ be executed before protocol $\bQ$ starts, $\bA^{(1)}_\PlayerB$ during $\bQ$, and
$\bA^{(2)}_\PlayerB$ after $\bQ$. Since $\bA^{(1)}_\PlayerB$ is $\min(m_1,m_2)$-bounded and $\bQ$ is secure, there exist a $s_1$-bounded $\bS^{(1)}_\PlayerB$, such that
$(\bQ_{\PlayerA},\bA^{(1)}_{\PlayerB})(\bH_{\{\PlayerB\}}) \equiv_{\eps_2} \bS^{(1)}_{\PlayerB}(\bG_{\{\PlayerB\}})$.
Let $\bA'_\PlayerB$ be the program that results from joining $\bA^{(0)}_\PlayerB$, $\bS^{(1)}_\PlayerB$, and $\bA^{(2)}_\PlayerB$. Because of $\max(\min(m_1,m_2),s_1) \leq m_2$, $\bA'_\PlayerB$ is $m_2$-bounded and $\bP$ is secure, there exists a $s_2$-bounded $\bS_\PlayerB$, such that
$(\bP_{\PlayerA},\bA'_{\PlayerB})(\bG_{\{\PlayerB\}} \| \bG'_{\{\PlayerB\}}) \equiv_{\eps_1} \bS_{\PlayerB}(\bF_{\{\PlayerB\}})$.
It follows now from the properties of the trace distance that $\bS_{\PlayerB}$ is a simulator that satisfies the security condition
for $\PlayerA$ with an error of at most $\eps_1 + \eps_2$.
The security for $\PlayerB$ can be shown in the same way.
\end{proof}

\paragraph{Interactive functionalities.}

The definitions above only apply to non-interactive functionalities, i.e. functionalities that consist of just one input/output phase.
In general, we would also like to securely implement functionalities with several such phases. 
The most prominent example of such a functionality is \emph{bit-commitment}, which has two phases, a \emph{commit-phase}, and an \emph{open-phase}.

The security definitions and the composition theorem generalize to the multi-phase case. Basically, 
all phases by themselves can be treated as individual, non-interactive functionalities, using the security definition
given above.
We can assume that the adversary always sends his internal classical and quantum state to the environment at the end of each phase, and
receives it back at the beginning of the next phase. The adversary can thus be modeled by individual
adversaries for each phase. However, since the ideal functionalities between the different phases are connected by some shared memory,
i.e., the actions of the functionality in the second phase may depend on the actions in the first phase,
the simulator must be allowed to use some \emph{classical} memory \emph{between} the rounds.

\section{Randomized Oblivious Transfer}\label{OTsecurity}

We now apply our framework to the randomized OT protocol 
presented in~\cite{serge:bounded}. In particular, we prove security
with respect to the following definition of randomized oblivious transfer.
We show in the appendix how to obtain the standard notion of
OT from randomized OT. Note that in our version of randomized OT,
also the choice bit $c$ of the receiver is randomized.

\begin{definition}[Randomized oblivious transfer] \label{def:rot}
$\ROT{1}{2}{\ell}$ (or, if $\ell$ is clear from the context, $\ROTT$) is defined as
$\ROTT = (\ROTT_\emptyset, \ROTT_{\{\PlayerA\}}, \ROTT_{\{\PlayerB\}})$, where
\begin{itemize}
\item $\ROTT_\emptyset$: The functionality chooses uniformly at random the value
$(x_0,x_1) \in_R \{0,1\}^{\ell} \times \{0,1\}^{\ell}$ and $c \in_R \{0,1\}$. It sends
$(x_0,x_1)$ to $\PlayerA$ and $(c,y)$ to $\PlayerB$ where $y = x_c$.
\item $\ROTT_{\{\PlayerA\}}$: The functionality receives $(x_0,x_1) \in \{0,1\}^{\ell} \times \{0,1\}^{\ell}$ from $\PlayerA$.
Then, it chooses $c \in_R \{0,1\}$ uniformly at random and sends
$(c,y)$ to $\PlayerB$, where $y = x_c$.
\item $\ROTT_{\{\PlayerB\}}$:  The functionality receives 
$(c,y) \in \{0,1\} \times \{0,1\}^{\ell}$ from $\PlayerB$. Then, it sets $x_c = y$, chooses $x_{1-c} \in_R \{0,1\}^{\ell}$ uniformly at random, and sends $(x_0,x_1)$ to $\PlayerA$.
\end{itemize}
\end{definition}

We first briefly recall the protocol. 
The protocol $\bqsOT = (\bqsOT_{\PlayerA}, \bqsOT_{\PlayerB}) $ uses a noiseless unidirectional quantum channel
$\qComm$, and a noiseless unidirectional classical channel $\comm$, both from the sender to the receiver.
Let $h: \mR \times \{0,1\}^n \rightarrow \{0,1\}^\ell$ be a two-universal hash function. 
The sender ($\PlayerA$) and receiver ($\PlayerB$) execute the following:

\begin{protocol}{$\bqsOT_{\PlayerA}$}{}
\item[1.] Choose $x \in_R \{0,1\}^n$ and $b \in_R \{0,1\}^n$ uniformly at random.
\item[2.] Send $\ket{x}_b := (\ket{x_1}_{b_1}, \dots, \ket{x_n}_{b_n})$ to $\qComm$, where $\ket{x_i}_{b_i}$ is $x_i$ encoded in the basis $b_i$.
\item[3.] Choose $r_0, r_1 \in_R \mR$ uniformly at random and send $(b, r_0, r_1)$ to $\comm$.
\item[5.] Output $(s_0,s_1) := (h(r_0,x_{|0}),h(r_1,x_{|1}))$, where $x_{|j}$ is the string of all $x_i$ where
$b_i = j$.
\end{protocol}

\begin{protocol}{$\bqsOT_{\PlayerB}$}{}
\item[1.] Choose $c \in_R \{0,1\}$ uniformly at random.
\item[2.] Receive the qubits $(q_1,\dots,q_n)$ from $\qComm$ and measure them in the basis $c$,
which gives output $x' \in \{0,1\}^n$.
\item[3.] Receive $(b, r_0, r_1)$ from $\comm$.
\item[4.] Output $(c,y) := (c,h(r_c,x'_{|c}))$, where $x'_{|c}$ is the string of all $x'_i$ where
$b_i = c$.
\end{protocol}

Note that the values $x_{|0}$, $x_{|1}$ and $x'_{|c}$ are in fact padded by additional $0$ to have a length
of $n$ bits. This padding does not affect their entropies. A memory bound is applied before step 1, and before step 3
of the receiver.

\paragraph{Security against the sender.}

We first consider the case when the sender, $\PlayerA$, is dishonest. 
This case turns out
to be quite straightforward. 
In general, we can describe any action of
the adversary by a unitary followed by a measurement in the computational basis. We use the following letters to refer to the different
classical and quantum registers available to the adversary: Let $\regQ$ denote the quantum register. Note that since
we assume that our adversary's memory is $m$-bounded, the size of $\regQ$ does not exceed $m$. Let $\regM_Q$ and $\regM_K$ denote
the quantum and classical registers, that hold the messages sent to the receiver. 
Let $\regK$ denote the classical input register of the adversary. 
Finally, let $\regA$ denote an auxiliary quantum register. 
Recall from Section~\ref{prelim}, that any quantum
operation on $\regQ$ and $\regM_Q$ can be implemented by a unitary followed by a measurement on an additional register $\regA$. 
Wlog we let $\regA$ and $\regM_Q$ be measured in the
computational basis to enforce a memory bound.

To model quantum and classical input that a malicious $\PlayerA$ may receive, we let $\regQ$ start out
in any state $\rho_{\inn}$, unknown to the simulator. Likewise, $\regK$ may contain some classical input $k_{\inn}$
of $\PlayerA$. Wlog we assume that all other registers start out in a fixed state
of $\ket{0}$.
We can then describe the actions of $\PlayerA$ by a single unitary $\bA_{\PlayerA}$ defined by
\begin{equation}\label{advTransform}
\bA_{\PlayerA}(\underbrace{\rho_{\inn}}_{\regQ} \otimes \underbrace{\outp{0}{0}}_{\regA} 
\otimes \underbrace{k_{\inn}}_{\regK} \otimes \underbrace{\outp{0}{0}}_{\regM_Q}
\otimes \underbrace{\outp{0}{0}}_{\regM_K}) 
\bA_{\PlayerA}^\dagger =
\underbrace{\rho_{\outt}}_{\regQ,\regA} \otimes \underbrace{k_{\inn}}_{\regK}
\otimes \underbrace{\rho_{x_b}}_{\regM_Q} \otimes \underbrace{\outp{b,r_0,r_1}{b,r_0,r_1}}_{\regM_K}.
\end{equation}
Note that without loss of generality $\bA_{\PlayerA}$ leaves $\regK$ unmodified: since $\regK$ is classical we can always
copy its contents to $\regA$ and let all classical output be part of $\regA$.
To enforce the memory bound, assume wlog that $\regA$ and $\regM_Q$ are now measured completely in the computational basis.
We now show that for any adversary $\bA_{\PlayerA}$ there exists an appropriate simulator $\bS_{\PlayerA}$.

\begin{lemma} \label{lem:secA}
Protocol $\bqsOT$ is secure against dishonest $\PlayerA$.
\end{lemma}

\begin{proof}
Let $\bS_{\PlayerA}$ be defined as follows: $\bS_{\PlayerA}$ runs
$\bA_{\PlayerA}$\footnote{As described in Section~\ref{BQS-Model}, $\bS_{\PlayerA}$ can effectively skip the wait time required for the memory bound to take effect, since
he can execute $\bA_{\PlayerA}$ before his memory bound is applied.}, and
measures register $\regM_Q$ in the basis determined by $\regM_K$.
This allows him to compute $s_0 = h(r_0,x_{|0})$ and $s_1 = h(r_1,x_{|1})$. 
$\bS_{\PlayerA}$ then sends $s_0$ and $s_1$ to
$\ROTT_{\{\PlayerA\}}$. It is clear that since the simulator based his measurement on $\regM_K$, 
$s_0$ and $s_1$ are consistent with the run of the protocol. 
Furthermore, note that $\bS_{\PlayerA}$ did not need to touch register $\regQ$ at all.
We can thus immediately
conclude that the environment can tell no difference between the real protocol and the ideal setting.
\end{proof}

\paragraph{Security against the receiver.}

To prove security against a dishonest receiver requires a more careful treatment of the quantum
input given to the adversary. The main idea behind our proof is that the memory bound in fact
\emph{fixes} a classical bit $c$. Our main challenge
is to find a $c$ that the simulator can calculate and that is consistent with the adversary and his input, while
keeping the output state of the adversary intact. To do so, we use a generalization of
the \emph{min-entropy splitting lemma} in \cite{serge:new}, which in turn
is based on an earlier version of \cite{Wullsc07}. It states that if two
random variables $X_0$ and $X_1$ together have high min-entropy, than we can define
a random variable $C$, such that $X_{1-C}$ has least half of the original min-entropy.
To find $C$, one must know the distributions of $X_0$ and $X_1$.
In the following generalization, we do \emph{not} exactly know the distribution of $X_0$ and $X_1$,
since we assume that its distribution also depends on an unknown random variable $J$, distributed over a domain
of the size $2^\beta$. 
Note that $\beta = 0$ give the min-entropy splitting lemma in \cite{serge:new}.

\begin{lemma}[Generalized Min-Entropy Splitting Lemma]\label{lem:genSplit}
Let $\eps \geq 0$, and $0 < \beta < \alpha$. Let $J$ be a random variable over $\{0,\dots,2^\beta-1\}$,
and let $X_0$, $X_1$ and $K$ be random variables such that 
$H^\eps_{\min}(X_0 X_1 \mid K J) \geq \alpha$.
Let $f(x_0,x_1,k) = 1$, if there exists an $j \in \{0,\dots,2^\beta-1\}$ such that
$P_{X_1 \mid K J}(x_1,k,j) \geq 2^{-(\alpha - \beta) / 2}$, and $0$ otherwise, and let
$C := f(X_0,X_1,K)$. We have
\[H^\eps_{\min}(X_{1-C}C \mid K J) \geq \frac{\alpha - \beta} 2\;.\]
\end{lemma}

\begin{proof}
Let $S^j_k$ be the set of values $x_1$ for which $P_{X_1 \mid K J}(x_1,k,j) \geq 2^{-(\alpha - \beta) / 2}$.
We have $|S^j_k| \leq 2^{(\alpha - \beta) / 2}$, since all values in $S^j_k$ have a probability that is at least $2^{-(\alpha - \beta) / 2}$.
Let $S_k := \bigcup_j S^j_k$. We have $|S_k| \leq 2^{\beta} \cdot 2^{ (\alpha - \beta) / 2} = 2^{(\alpha + \beta) / 2}$.

Let $K = k$ and $J = j$. 
Because $C=0$ implies that $X_1 \not \in S_k$, and thus also that $X_1 \not \in S^j_k$,
we have $P_{X_1 C \mid K J} (x_1,0,k,j) < 2^{(\alpha - \beta) / 2}$.
It follows from the assumption that there exists an event $\mE$ with probability $1- \eps$ such that for all $x_0$, $x_1$, $k$ and $j$, we have
$P_{X_0 X_1 \mE \mid K J} (x_0,x_1,k,j) \leq 2^{-\alpha}$. It follows that
\[P_{X_0 C \mE \mid K J} (x_0,1,k,j) = \sum_{x_1 \in S_k}  P_{X_0 X_1 \mE \mid K J} (x_0,x_1,k,j) \leq 2^{(\alpha + \beta) / 2} \cdot 2^{-\alpha} = 2^{-(\alpha - \beta) / 2}\;.\]
The statement follows.
\end{proof}

We now describe the actions of the adversary. Let $\regQ$ denote his quantum storage
register, and let $\regA$ denote an auxiliary quantum register as above.
Again, the size of $\regQ$ does not exceed $m$.
Let $\regK$ denote his classical input register, and 
let $\regM$ denote the register holding the 
quantum message he receives from the sender in step 2.
Let $\regE$ denote the message register
holding the classical messages he receives in step 3. 
Again, we assume that $\regQ$ is initialized
to his quantum input state $\rho_{\inn}$. 
Likewise, $\regK$ is initialized to his classical input $k_{\inn}$. All other registers
are initialized to $\ket{0}$. We can now describe the actions of the adversary by two unitaries, 
where a memory bound is applied after the first. The action of the adversary following step 2 can be described as
a unitary $\bA^{(1)}_{\PlayerB}$ similar to Eq.~\ref{advTransform}.
Note we can again assume that $\bA^{(1)}_{\PlayerB}$ leaves $\regK$ unmodified. 
To enforce the memory bound, we now let register $\regM$ and $\regA$ be measured in the computational basis.
We use $\rho_{\outt} \in \regQ$ to denote the adversaries quantum output, and $k_{\outt} \in \regM \otimes \regA$
to denote his classical output. 
After the memory bound is applied, the receiver obtains additional information from the sender. The actions of the
adversary after step 3 can then be described by a unitary $\bA^{(2)}_{\PlayerB}$ followed by a measurement
of quantum registers $\regM$ and $\regA$ in the computational basis.

In order to make the proof easier to understand, we build it up in 3 steps: First, we analyze the easy case where 
there is no quantum auxiliary input, which is essentially equivalent to the original security proof. Then
we extend it, by allowing the adversary some quantum auxiliary input of size $\beta$, pure and mixed.
We start with $\beta = 0$, but  keep $\beta$ as a parameter, so that we can later generalize
the statement.

\begin{lemma} \label{lem:secB-noAux}
Protocol $\bqsOT$ is secure against dishonest $\PlayerB$ with an error of at most $5\eps$,
if he receives no quantum (auxiliary) input, and his quantum memory is bounded
before step 1 and between step $2$ and $3$ by $m$ qubits, for
$$8 \ell + 2 \beta + 4 m \leq n - 20 \sqrt[3]{n^2 \log{\frac{1}{\eps}}} - 12 \log{\frac{1}{\eps}} - 4,$$
where $\beta$ is a parameter.
\end{lemma}

\begin{proof}
Let $K_{\inn}$ be the classical auxiliary input the adversary receives, and let
 $\ket{\Psi_{\inn}} = \ket{j}$ be the auxiliary quantum input for some fixed $j$ known to the simulator. 
First of all, the simulator simulates the actions of the sender following steps 1 and 2, using a random string
$X$ and a random basis $B$. The simulator then applies $\bA^{(1)}_{\PlayerB}$,
which gives him some classical output $K_{\outt}$, and a quantum state $\rho_{\outt}$.
It follows from the uncertainty relation of Lemma~\ref{lem:uncertainty2} that
$$
H^{2\eps}_{\min}(X \mid B K_{\outt} K_{\inn}) \geq \alpha\;,
$$
for $\alpha := n/2 - 10 \sqrt[3]{n^2 \log(1/\eps)}$. Let $(X_0,X_1) := X$, where $X_0:=X_{|0}$ and $X_1:=X_{|1}$ are 
the substrings of $X$ defined in the same way as in the protocol.
Note that since the simulator holds a description of $\ket{\Psi_{\inn}}$ and $\bA^{(1)}$, he knows
$P_{X_0 X_1 B K_{\outt} K_{\inn}}$, and thus we can apply Lemma~\ref{lem:genSplit},
for $K = (B,K_{\outt},K_{\inn})$ and a constant $J$ (or, $\beta = 0$).
Since the simulator knows the values $X_0$, $X_1$ and $K$, he can
calculate the value $C := f(X_0, X_1, K)$, for which we have
$$
H^{\eps}_{\min}(X_{1-C}C \mid K) \geq \frac{\alpha - \beta} 2\;.
$$
The simulator now chooses $R_0$ and $R_1$ uniformly at random and calculates
$S_0 = h(R_0,X_0)$ and $S_1 = h(R_1,X_1)$. Since $R_0$ and $R_1$ are independent of $X_0$, $X_1$ and $C$, we have
\[ H^{\eps}_{\min}(X_{1-C}C \mid K) =
H^{\eps}_{\min}(X_{1-C}C \mid R_C K)\;.\] 
Using the
chain rule from Lemma~\ref{lem:chain} and the monotonicity of $H^\eps_{\min}$, we obtain
\begin{eqnarray*}
H_{\min}^{2\eps}(X_{1-C} \mid C R_C K S_C)
&\geq& H_{\min}^{\eps}(X_{1-C} S_C C \mid R_C K)
- (\ell + 1) - \log{\frac{1}{\eps}}\\
&\geq& H_{\min}^{\eps}(X_{1-C} C \mid R_C K)
- (\ell + 1) - \log{\frac{1}{\eps}}\\
&\geq& \frac{\alpha - \beta} 2 - \ell - 1 - \log{\frac{1}{\eps}}\;.
\end{eqnarray*}
By using the privacy amplification Theorem~\ref{thm:QBS-ILL}, we get that $S_{1-C}$ is $5 \eps$ close
to uniform with respect to $(R_0,R_1,C,S_C,B,K_{\outt},K_{\inn})$ and $\rho_{\outt}$ if 
$$
\ell \leq \frac{\alpha - \beta} 2 - \ell - 1 - \log{\frac{1}{\eps}} - m - 2 \log{\frac{1}{\eps}}\;.
$$
By replacing $\alpha$ and rearranging the terms we get the claimed equation.

The simulator now sets $Y := S_C$, 
and sends $(C,Y)$ to $\ROTT_{\{\PlayerB\}}$. To complete the simulation, he runs
$\bA^{(2)}_{\PlayerA}$ as the adversary would have.
Note that the simulator did not require any more memory than the adversary itself, i.e., 
we can take $\bS_\PlayerB$ to be $m$-bounded as well.
Clearly, the simulator determined $C$ solely from the classical output of the adversary and
thus the adversaries output state in the simulated run is equal to the original output 
state of the adversary $\rho_{\outt} \otimes k_{\outt}$. Since the only difference between
the simulation and the real execution is that in the simulation, $S_{1-C}$ is chosen completely at random,
the simulation is $5 \eps$-close to the output of the real protocol.
\end{proof}

We now show how to extend the above analysis to the case where the adversary's input is pure. Note that
if the adversary's input is pure, the adversary cannot be entangled with the environment.

\begin{lemma} \label{lem:secB-withAux}
Protocol $\bqsOT$ is secure against dishonest $\PlayerB$ with an error of at most $5\eps$,
if he receives a pure state quantum (auxiliary) input, and his quantum memory is bounded
before step 1 by $\beta$ qubits, and between step $2$ and $3$ by $m$ qubits, for
$$8 \ell + 2 \beta + 4 m \leq n - 20 \sqrt[3]{n^2 \log{\frac{1}{\eps}}} - 12 \log{\frac{1}{\eps}} - 4.$$
\end{lemma}

\begin{proof}
Let $\ket{j}$ for $j \in \{0, \dots, 2^{\beta}\}$ be a basis for the quantum auxiliary input.
Any fixed auxiliary input $\ket{j}$ and $k_{\inn}$ fixes a distribution $P_{X_0X_1K \mid J=j}$, where $K$
is the classical value the adversary has after second memory bound. Using the same argumentation as in
Lemma \ref{lem:secB-noAux}, but now using the generalized min-entropy splitting lemma with $\beta > 0$,
we can construct a simulator that does not need to know $j$, nor the distribution $P_J$.

Hence, the simulator can construct a 
linear transformation acting on registers $\regQ$, $\regM$, $\regA$, $\regK$, 
$\regX$, $\regB$, $\regR$, and $\regC$ combing the actions of $\bA_\PlayerA^{(1)}$ and the extraction of $c$
using the function $f$ as defined above. We have
\begin{eqnarray*}
&&\bS^{2}_{\PlayerB}(
\sum_j \alpha_j 
\underbrace{\ket{j}}_\regQ \otimes
\underbrace{\ket{x_b}}_\regM \otimes \underbrace{\ket{0}}_\regA
\otimes \underbrace{\ket{k_{\inn}}}_\regK \otimes \underbrace{\ket{x}}_\regX \otimes 
\underbrace{\ket{b}}_\regB \otimes \underbrace{\ket{r_0,r_1}}_\regR \otimes
\underbrace{\ket{0}}_\regC \otimes \underbrace{\ket{0}}_\regY) =\\
&&\sum_{q,m_1,a_1} \alpha_{q,m_1,a_1}
\underbrace{\ket{q}}_\regQ \otimes
\underbrace{\ket{m_1}}_\regM \otimes \underbrace{\ket{a_1}}_\regA
\otimes \underbrace{\ket{k_{\inn}}}_\regK \otimes \underbrace{\ket{x}}_\regX \otimes 
\underbrace{\ket{b}}_\regB \otimes \underbrace{\ket{r_0,r_1}}_\regR \otimes
\underbrace{\ket{c}}_\regC \otimes \underbrace{\ket{s_0,s_1}}_\regY)
\end{eqnarray*}
for any pure state input $\ket{\Psi_{\inn}} = \sum_j \alpha_j \ket{j}$.
Wlog, all registers except $\regQ$ are now measured in the computational basis as the memory bound takes effect.
The input state $\ket{\Psi_{\inn}}$ will define the distribution of $J$ for the generalized min-entropy splitting
lemma. The rest follows as above.
\end{proof}

It remains to address the case where the receiver gets a mixed state quantum input. This is
the case where the adversary receives a state that is entangled with the environment. Note that this means
that we must decrease the size of the adversaries memory: If he could receive an entangled state of
$\beta$ qubits as input, he could use it to increase his memory to $m + \beta$ qubits by teleporting $\beta$ qubits
to the environment, and storing the remaining $m$. Hence, we now have to take the adversary to
be $m'$-bounded, where $m' := m - \beta$. Luckily, using a
a similar argument as in~\cite{watrous:zks}, we can now extend the argument given above: 
Note that for any pure state input 
$\ket{\Psi} = \ket{\Psi_{\inn}} \otimes k_{\inn}$, the output of the simulated adversary
is \emph{exactly} 
$\Lambda(\outp{\Psi}{\Psi})$, where $\Lambda$ is the adversaries channel.
Since $\{\outp{\Psi}{\Psi}| \ket{\Psi} \in \regQ \otimes \regK, \norm{\ket{\Psi}} = 1\}$ spans all
of $\setT(\regQ \otimes \regK)$ and the map given by the simulation procedure is the same
as $\Lambda$ on all inputs, we can conclude that the complete map is equal to $\Lambda$. 
Note that the simulator does not need to consider the $\beta$ qubits that the adversary might
have teleported to the environment: we can essentially view it as part of the original adversaries
quantum memory, and the simulator bases his decision solely on the classical output
of the adversary. Hence:

\begin{lemma} \label{lem:secB-mixedAux}
Protocol $\bqsOT$ is secure against dishonest $\PlayerB$ with an error of at most $5\eps$,
if he receives a quantum (auxiliary) input, and his quantum memory is bounded
before step 1 by $\beta$ qubits and between step $2$ and $3$, by $m$ qubits, for
$$8 \ell + 6 \beta + 4 m \leq n - 20 \sqrt[3]{n^2 \log{\frac{1}{\eps}}} - 12 \log{\frac{1}{\eps}} - 4.$$
\end{lemma}

The following theorem follows now directly from Lemma~\ref{lem:secA} and~\ref{lem:secB-mixedAux}.

\begin{theorem} \label{thm:bqsOT}
Protocol $\bqsOT(\qComm \| \comm)$ implements $\ROT{1}{2}{\ell}$ with an error of at most
$5\eps$, secure against $m$-bounded adversaries using $m$-bounded simulators, if
$$
8 \ell + 10 m \leq n - 20 \sqrt[3]{n^2 \log{\frac{1}{\eps}}} - 12 \log{\frac{1}{\eps}} - 4.
$$
\end{theorem}

Note that there are ways to improve on these parameters. For example, the splitting lemma
defines the function $f$ in an asymmetric way, which implies that for $C=1$, in fact the bound also holds for the
\emph{conditional} min-entropy of $X_0$ given $X_1$. Thus, we would not need to additionally
apply the chain rule for this case. We did not do this here to keep the proof simple.

\subsection{On Parallel Composition}

For efficiency, it would be important to know if the protocol from last section would also be secure under
 \emph{parallel} composition. Unfortunately, this is not an easy task:
First, consider 
executing the protocol in parallel, when the sender and receiver are the same for
each instance of the protocol. Clearly, the overall memory of the committer cannot exceed the amount of memory he
would be allowed for a single execution of the protocol: otherwise he could cheat in at least one instance of the protocol.
However, even when imposing such a constraint, parallel composition remains tricky: Second, consider the case where
we run two instances of the protocol in parallel, where the roles of the sender (initially Alice) and the 
receiver (initially Bob) are \emph{exchanged} in the second instance of the protocol. Let the malicious Bob behave
as follows: Upon reception of the quantum states in the first instance of the protocol, he immediately returns them 
unmeasured to Alice. Later, he sends the very same values $(b, r_0, r_1)$ back to Alice. Alice thus measures her
own states, in her own bases. Thus, her output $y$ of the second instance of the protocol will always be equal
either to $x_0$ or to $x_1$ of the first instance of the protocol. This is clearly something Bob would not be able to
do in an ideal setting. This simple example already shows that great care must be taken when composing such
protocols in parallel: no quantum memory was required to execute such an attack.

\section{Acknowledgments}
We thank Simon Pierre Desrosiers and Christian Schaffner for useful comments, and Dominique Unruh for
a kind explanation of his work.

\appendix

\section{Proof of Lemma \ref{lem:uncertainty2}}

To prove Lemma \ref{lem:uncertainty2}, we need Lemma~\ref{lem:bound} below and the following Theorem~\ref{thm:uncertainty}, which is Corollary 3.4 in the full version of \cite{serge:new}.

\begin{theorem}[Uncertainty Relation \cite{serge:new}] \label{thm:uncertainty}
Let $\rho \in \setS(\mH^{\otimes n}_2 )$ be an arbitrary quantum state. Let $\Theta = (\Theta_1,\dots, \Theta_n)$ be uniformly distributed
over $\{+,\times\}$ and let $X = (X_1, \dots ,X_n)$ be the outcome when measuring $\rho$ in basis $\Theta$. Then for any $0 < \lambda < \frac12$
\[ H^{\eps}_{\min}(X |\Theta) \geq (\frac12 - 2\lambda) n \]
with
$\eps = \exp \left ( - \frac{\lambda^2 n}{32  (2 - \log(\lambda))^2} \right )$.
\end{theorem}

\begin{lemma} \label{lem:bound}
For all $0 < x \leq 0.5$, we have
\[\frac{1}{(2 - \log(x))^{2}} \geq  \frac { e^3 \ln(2)^2} {54} \cdot  x \;.\]
\end{lemma}

\begin{proof}
Since $(2 - \log(x))^{-2} \rightarrow 0$ for $x \rightarrow 0$, it
suffices to show that, for all $0 < x \leq 0.5$,
\[\frac{d}{dx} \frac{1}{(2 - \log(x))^{2}} = \frac{2}{\ln(2)}
\frac{1}{(2 - \log(x))^{3} x} \geq \frac { e^3 \ln(2)^2} {54}\;,\]
with is equivalent to require that
\[f(x) := \frac {\ln(2)}{2} (2 - \log(x))^{3} x \leq \frac{54} { e^3
\ln(2)^2}\;.\]
We have
\[f'(x) = -3 (2 - \log(x))^2 + (2 - \log(x))^3 \ln(2)\;,\]
and since the polynomial $-3 x^2 + x^3 \ln(2)$ has a double root at
$0$ and a single root at $3/\ln(2)$, and is positive if and only if $x > 3/\ln(2)$, it follows that
$f(x)$ has one double root at $4$, and one single root at $4/e^3$. It
is positive for $0 < x<4/e^3$ and negative for
$4/e^3 < x \leq 0.5$. Hence, $f(x)$ is maximal for $x=4/e^3$, where
$f(4/e^3) = 54 / ( e^3 \ln(2)^2)$.
\end{proof}

\begin{proof}[Proof of Lemma \ref{lem:uncertainty2}]
Following the standard approach (see also~\cite{serge:new}), we consider a purified
version of our situation: Alice creates $n$ EPR pairs, and sends the second half of each pair to Bob.
His measurement is then applied onto the second half of these pairs, which has output $K$.
Then, we choose uniform a random basis $\Theta \in \{+,\times\}^n$, and measure the first half in this basis, which gives us
the string $X$. The output of the purified situation is identical to the situation in the statement, however it allows us
to apply Corollary~\ref{thm:uncertainty}.

From $10\sqrt[3]{n^2  \log(1/\eps)} = n/2 \sqrt[3]{8000 \log(1/\eps)/n}$ follows that
$n/2 > 10 \sqrt[3]{n^2  \log(1/\eps)}$ if and only if $n > 8000 \log(1 / \eps)$.
Thus, nothing has to be shown if $n \leq 8000  \log(1 / \eps)$.
If $n > 8000 \log(1 / \eps)$, we
choose $\lambda :=  5 \sqrt[3]{1/n \cdot \log(1/\eps)}$  and $\lambda' := 1/n \cdot \log(1/\eps)$, and get
\[
\lambda =  5 \sqrt[3]{\frac 1 n \cdot \log(1/\eps)} < 5 \sqrt[3]{\frac 1 { 20^{3} \cdot \log(1 / \eps)} \cdot \log(1/\eps)} = 5 \sqrt[3]{ \frac 1 {{20}^{3}}} = \frac 1 4\;.\]
The statement follows from
\begin{align*}
\exp \left ( - \frac{\lambda^2 n}{32 \cdot (2 - \log(\lambda))^2} \right )
\leq \exp \left ( - \frac{\lambda^3 n }{180} \right )
= 2^{- \frac{\log(e)}{180}  \lambda^3 n}
\leq 2^{- \frac{1}{5^3} \lambda^3 n}
\leq 2^{- \log(1/\eps)}
= \eps\;.
\end{align*}
\end{proof}

\section{Oblivious Transfer from ROT}

Oblivious transfer is defined as follows:

\begin{definition}[Oblivious transfer]
The functionality $\OT{1}{2}{\ell}$ receives input $(x_0,x_1) \in \{0,1\}^{\ell} \times \{0,1\}^{\ell}$ from
$\PlayerA$ and $c \in \{0,1\}$ from
$\PlayerB$, and sends
$y := x_c$ to $\PlayerB$.
\end{definition}

The protocol $\OTfromROT$, proposed in
\cite{BBCS92}, securely implements $\OTT$ using $\ROTT$ and $\comm$.

\begin{protocol}{$\OTfromROT_{\PlayerA}$}{}
\item Receive input $(x_0,x_1) \in \{0,1\}^{\ell} \times \{0,1\}^{\ell}$ and
$(x'_0,x'_1) \in \{0,1\}^{\ell} \times \{0,1\}^{\ell}$ from $\ROTT$.
\item Receive $d \in \{0,1\}$ from $\comm$.
\item Send $(m_0,m_1) \in \{0,1\}^{\ell} \times \{0,1\}^{\ell}$
to $\comm$, where $m_i := x_i \oplus x'_{i \oplus d}$.
\end{protocol}

\begin{protocol}{$\OTfromROT_{\PlayerB}$}{}
\item Receive input $c \in \{0,1\}$ and $(c',y') \in
 \{0,1\} \times \{0,1\}^{\ell}$ from $\ROTT$.
\item Send $d := c' \oplus c$ to $\comm$.
\item Receive $(m_0,m_1) \in \{0,1\}^{\ell} \times \{0,1\}^{\ell}$ from $\comm$ and output $y := m_c \oplus y'$ to $\PlayerB$.
\end{protocol}

\begin{theorem} \label{thm:ROTfromOTactive}
For every $m > 0$, $\OTfromROT(\ROT{1}{2}{\ell}\|\comm)$ implements $\OT{1}{2}{\ell}$ with no error,
secure against $m$-bounded adversaries using $m$-bounded simulators.
\end{theorem}

\begin{proof}
It is easy to verify that the protocol is correct, if $\setA =  \emptyset$.

Let $\setA = \{\PlayerA\}$ and $\bA_{\PlayerA}$ be a quantum adversary. $\bA_{\PlayerA}$ receives some auxiliary input\footnote{Now, auxiliary inputs and output are always both classical and quantum.}, and outputs $(x'_0,x'_1)$ which are the inputs to $\ROTT_{\PlayerA}$. Then it receives $d$, and finally output $(m_0,m_1)$ and some auxiliary output.
The simulator $\bS_\PlayerA$ works as follows. It receives some auxiliary input, and then executes 
$\bA_{\PlayerA}$, using the auxiliary input. It stores the values $(x'_0,x'_1)$ returned by $\bA_{\PlayerA}$,
and sends it a value $d$ chosen uniformly at random.
$\bA_{\PlayerA}$ then outputs $(m_0,m_1)$ and some
auxiliary output. The simulator outputs the auxiliary output and
sends the values $x_i := m_i \oplus x'_{i \oplus d}$ for $i \in \{0,1\}$ to $\OTT$. It is easy to verify that the real and
the simulated situations give exactly the same output distribution.

Let $\setA = \{\PlayerB\}$ and $\bA_{\PlayerB}$ be a quantum adversary. 
$\bA_{\PlayerB}$ receives some auxiliary input and outputs $(c'_0,y')$, which are the inputs to $\ROTT_{\PlayerA}$, and a value
$d$. Then it receives the values $(m_0,m_1)$, and outputs some auxiliary output.
The simulator works as follows. It receives some auxiliary input, and then executes 
$\bA_{\PlayerB}$ on the auxiliary input, which outputs $(c'_0,y')$ and a value $d$.
The simulator now sends $c := c' \oplus d$ to $\OTT$, and receives a value $y$ back.
Then, it sets $m_{c' \oplus d} := y \oplus y'$, chooses the other value $m_{c' \oplus d \oplus 1}$
uniformly at random, and sends $(m_0,m_1)$ to $\bA_{\PlayerB}$. Finally, it outputs the auxiliary output
that $\bA_{\PlayerB}$ returns. It is easy to verify
that the real and the simulated situations give exactly the same output distribution.
\end{proof}

\section{Bit-Commitment from ROT}

In \cite{serge:bounded}, a bit-commitment protocol is presented and proved secure for a weak binding condition. 
In \cite{serge:new}, it is shown that the same protocol is in fact also secure under a stronger binding condition. However,
as for $\ROTT$, their proof does not take auxiliary inputs into account.
In a similar way as $\ROTT$, their protocol could be proven secure in our framework.
But because protocols can be composed in our framework, we can now give a much simpler proof: 
We can implement bit-commitment based directly on $\ROTT$.
The composition theorem then implies that if $\ROTT$ is 
replaced by an instance of $\bqsOT$, the bit-commitment protocol remains secure.
The $\BC$ functionality is defined as follows.

\begin{definition} \label{def:bc}
The functionality $\BC$ has two phases, which are defined as follows:
\begin{itemize}
\item Commit: $\BC$ receives $b \in \{0,1\}$ from $\PlayerA$ and sends $\perp$ to $\PlayerB$.
\item Open: $\BC$ receives $a \in \{0,1\}$ from $\PlayerA$. If $a=1$, it sends $b$ to $\PlayerB$. Otherwise, it sends
$\perp$.
\end{itemize}
\end{definition}

Let $\TOR{1}{2}{\ell}$ be a reversed version of of $\ROT{1}{2}{\ell}$, i.e., $\PlayerB$ is the sender and $\PlayerA$ is the receiver.
The protocol $\OTtoBC = (\OTtoBC_{\PlayerA}, \OTtoBC_{\PlayerB})$ uses $\TOR{1}{2}{\ell}$ and a noiseless unidirectional classical $\comm$ from $\PlayerA$ to $\PlayerB$ to implement $\BC$. We now first describe the actions of the committer.

\begin{protocol}{$\OTtoBC_{\PlayerA}$}{}
\item[{\bf Commit:}] \ 
\begin{itemize}
\item[1.] Receive input $b$ from $\PlayerA$ and $(c,y) \in \{0,1\} \times \{0,1\}^n$ from $\TORR$.
\item[2.] Send $m := b \oplus c$ to $\comm$.
\end{itemize}
\item[{\bf Open:}] \ 
\begin{itemize}
\item[1.] Receive input $a$ from $\PlayerA$. If $a=1$, then send $(b,y)$ to $\comm$, and $(\perp,\perp)$ otherwise.
\end{itemize}
\end{protocol}

The actions of the verifier are specified by:

\begin{protocol}{$\OTtoBC_{\PlayerB}$}{}
\item[{\bf Commit:}] \ 
\begin{itemize}
\item[1.] Receive $(x_0,x_1)$ from $\TORR$ .
\item[2.] Receive $m$ from $\comm$ and output $\perp$.
\end{itemize}
\item[{\bf Open:}] \ 
\begin{itemize}
\item[1.] Receive $(b,y)$ from $\comm$.
\item[2.] If $(b,y) \neq (\perp,\perp)$ and $x_{b \oplus m} = y$, then output $b$, and $\perp$ otherwise.
\end{itemize}
\end{protocol}

\begin{theorem} \label{thm:OTtoBC}
For every $m > 0$, $\OTtoBC(\TOR{1}{2}{\ell} \| \comm)$ implements $\BC$ with an error of at most $2^{-\ell}$,
secure against $m$-bounded adversaries using $m$-bounded simulators.
\end{theorem}

\begin{proof}
It is easy to verify that the protocol is correct, if $\setA =  \emptyset$

Let $\setA = \{\PlayerA\}$ and $\bA_{\PlayerA}$ be a quantum adversary. In the commit phase, it receives some auxiliary input, sends $(c,y)$ to $\TORR$ and $m$ to $\comm$ and outputs some auxiliary output.
In the open phase, it receives some
auxiliary input, sends $(b,y')$ to $\comm$, and outputs some auxiliary output.
Note that in the real execution, if $(b,y') = (c \oplus m,y)$, the protocol outputs $b$ to $\PlayerB$ in the open phase. On the other hand, if $b \neq c \oplus m$, the protocol will only output $b$ to $\PlayerB$,  if $y = x_{1 - c}$.
Since $x_{1 - c}$ is chosen uniformly at random, this will only happen with a probability of $2^{-\ell}$.
The simulator $\bS_{\PlayerA}$ does the following: In the commit phase, it receives the auxiliary input, and sends it to
$\bA_{\PlayerA}$ to run the commit phase, from which it receives $(c,y)$, $m$, and some auxiliary output. It outputs the auxiliary output, sends $c \oplus m$ to $\BC$, and saves $c \oplus m$ in his classical memory.
In the open phase, it receives some auxiliary input and sends it to 
$\bA_{\PlayerA}$ to run the open phase. It receives $(b,y')$ from $\bA_{\PlayerA}$, and sends $a=1$ to $\BC$ if $b = c \oplus m$, and $a = 0$ otherwise. Finally, it outputs the auxiliary output of $\bA_{\PlayerA}$.
It is easy to verify that simulation is equal to the real execution, except with probability at most $2^{-\ell}$.

Let $\setA = \{\PlayerB\}$ and $\bA_{\PlayerB}$ be a quantum adversary.  In the commit phase, it receives some
auxiliary input, sends $(x_0,x_1)$ to $\ROTT$, receives $m$, and outputs some auxiliary output.
In the open phase, it receives some
auxiliary input and $(b,y)$ from $\comm$, and outputs some auxiliary output.
The simulator
$\bS_{\PlayerB}$ does the following: In the commit phase, it receives some auxiliary input, 
and sends it to $\bA_{\PlayerB}$, from which it receives $(x_0,x_1)$. Then, it sends a value $m$ chosen uniformly at random to $\bA_{\PlayerB}$, and gets some auxiliary output back. It outputs the
auxiliary output, and stores $(x_0,x_1)$ and $m$ in the classical memory. In the open phase, it receives some auxiliary input,
and a value $b'$ from $\BC$. If $b' = \perp$, it sets $y := \perp$, and $y := x_{b' \oplus m}$ otherwise. It sends the auxiliary input and  $(b',y)$ to $\bA_{\PlayerB}$, and outputs the auxiliary output returned by $\bA_{\PlayerB}$.
It is easy to see that the simulation produces always exactly the same output as the simulation.
\end{proof}

Let $\bqsTO$ be the same protocol as $\bqsOT$, but in the opposite direction.
Using Theorem~\ref{thm:bqsOT} and \ref{thm:OTtoBC}, as well as Theorem~\ref{thm:bqs-compose}, and by choosing $\ell := \log 1/\eps$, we get

\begin{theorem}
The Protocol $\OTtoBC( \bqsTO(\qComm \| \comm) \| \comm)$ implements $\BC$ with an error of at most
$6\eps$, secure against $m$-bounded adversaries using $m$-bounded simulators, if
$$
10 m \leq n - 20 \sqrt[3]{n^2 \log{\frac{1}{\eps}}} - 20 \log{\frac{1}{\eps}} - 4
$$
\end{theorem}

\end{document}